\documentclass[11pt,letterpaper]{article}

\usepackage[in]{fullpage}

\usepackage[utf8]{inputenc}
\usepackage[T1]{fontenc}
\usepackage[english]{babel}
\usepackage[hidelinks]{hyperref}       
\hypersetup{
	colorlinks,
	linkcolor={red!50!black},
	citecolor={blue!70!black},
}

\usepackage{url}            
\usepackage{booktabs}       
\usepackage{amsfonts}       
\usepackage{nicefrac}       
\usepackage{microtype}      
\usepackage{amsmath, amsthm, amssymb}
\usepackage{fancybox}
\usepackage{graphicx}
\usepackage{float}
\usepackage{algorithm}
\usepackage{color}
\usepackage{tikz}
\usepackage{pgfplots}
\usepackage{array}
\usepackage{subcaption}
\usepackage[textsize=scriptsize,disable]{todonotes}
\usepackage{enumitem}
\usepackage{wrapfig}
\usepackage{enumitem}
\usepackage{etoolbox}
\usepackage{diagbox}
\usepackage{cleveref}

\newtheoremstyle{slplain}
  {.4\baselineskip\@plus.1\baselineskip\@minus.1\baselineskip}
  {.3\baselineskip\@plus.1\baselineskip\@minus.1\baselineskip}
  {\itshape}
  {}
  {\bfseries}
  {.}
  { }
  {}
\theoremstyle{slplain} 

\newtheorem*{definition*}{Definition}
\newtheorem*{theorem*}{Theorem}
\newtheorem{theorem}{Theorem}[section]
\newtheorem{lemma}[theorem]{Lemma}

\newtheorem{corollary}[theorem]{Corollary}
\newtheorem{definition}[theorem]{Definition}

\makeatletter
\newtheorem*{rep@theorem}{\rep@title}
\newcommand{\newreptheorem}[2]{%
\newenvironment{rep#1}[1]{%
 \def\rep@title{#2 \ref{##1}}%
 \begin{rep@theorem}}%
 {\end{rep@theorem}}}
\makeatother

\newreptheorem{theorem}{Theorem}
\newreptheorem{lemma}{Lemma}
\newreptheorem{claim}{Claim}
\newreptheorem{corollary}{Corollary}
\newreptheorem{proposition}{Proposition}

\theoremstyle{definition}

\theoremstyle{plain} 

\numberwithin{equation}{section}

\newtheoremstyle{etplain}
  {.0\baselineskip\@plus.1\baselineskip\@minus.1\baselineskip}
  {.0\baselineskip\@plus.1\baselineskip\@minus.1\baselineskip}
  {\itshape}
  {}
  {\bfseries}
  {.}
  { }
  {}

\renewcommand\bar\overline
\renewcommand\epsilon\varepsilon

\usepackage[noend]{algorithmic}

\newcommand{\newpar}[1]{\medskip\noindent{\bf #1}\ \ }
\newenvironment{prevproof}[2]{\noindent {\bf {Proof of {#1}~\ref{#2}.}}}{$\hfill\qed$\vskip \belowdisplayskip}

\begin{document}

\title{
	Capacitated Dynamic Programming: \protect\\
	Faster Knapsack and Graph Algorithms
}

\author{
  Kyriakos Axiotis\\
  MIT\\
  \texttt{kaxiotis@mit.edu} 
  \and
  Christos Tzamos\\
  University of Wisconsin-Madison\\
  \texttt{tzamos@wisc.edu}
}
\date{}

\maketitle

\begin{abstract}

One of the most fundamental problems in Computer Science is the \emph{Knapsack} problem. Given a set of $n$ items with different weights and values, it asks to pick the most valuable subset whose total weight is below a capacity threshold $T$. Despite its wide applicability in various areas in Computer Science, Operations Research, and Finance, the best known running time for the problem is $O(T n)$. The main result of our work is an improved algorithm running in time $O(TD)$, where $D$ is the number of distinct weights. Previously, faster runtimes for Knapsack were only possible when both weights and values are bounded by $M$ and $V$ respectively, running in time $O(nMV)$ \cite{Pisinger99}. In comparison, our algorithm implies a bound of $O(n M^2)$ without any dependence on $V$, or $O(n V^2)$ without any dependence on $M$. Additionally, for the unbounded Knapsack problem, we provide an algorithm running in time $O(M^2)$ or $O(V^2)$.
Both our algorithms match recent conditional lower bounds shown for the Knapsack problem \cite{CMWW17,KPS17}.

We also initiate a systematic study of general \emph{capacitated dynamic programming}, of which Knapsack is a core problem. This problem asks to compute the maximum weight path of length $k$ in an edge- or node-weighted directed acyclic graph. In a graph with $m$ edges, these problems are solvable by dynamic programming in time $O(k m)$, and we explore under which conditions the dependence on $k$ can be eliminated. We identify large classes of graphs where this is possible and apply our results to obtain linear time algorithms for the problem of $k$-sparse $\Delta$-separated sequences. The main technical innovation behind our results is identifying and exploiting concavity that appears in relaxations and subproblems of the tasks we consider.
\end{abstract}

\section{Introduction}
A large number of problems in Computer Science can be formulated as finding the optimal subset of items to pick in order to maximize a given objective subject to capacity constraints.

A core problem in this class is the \emph{Knapsack problem}: In this problem, each of the $n$ items has a value and a weight and the objective is to maximize the total value of the selected items while having total weight at most $T$.

A standard approach for solving such capacitated problems is to use dynamic programming.
Specifically, the dynamic programming algorithm keeps a state that tracks how much of the available capacity has already been exhausted.
The runtime of these algorithms typically incurs a multiplicative factor equal to the total capacity.
In particular, in the case of the Knapsack problem the classical dynamic programming algorithm due to Bellman~\cite{Bellman57} has a runtime of $O(Tn)$.

In contrast, uncapacitated problems do not restrict the number of elements to be selected, but charge an extra cost for each one of them (i.e. they have a \emph{soft} as opposed to a
\emph{hard} capacity constraint). 
The best known algorithms for these problems are usually much faster than the ones for their capacitated counterparts, i.e. for the uncapacitated version of knapsack one would need to pick all items whose value is larger than their cost. Therefore a natural question that arises is whether or when the additional dependence of the runtime on the capacity is really necessary.

In this work, we make progress towards answering this question by exploring when this dependence can be improved or completely eliminated.

\newpar{Knapsack} We first revisit the Knapsack problem and explore under which conditions we can obtain faster algorithms than the standard dynamic programming algorithm.

Despite being a fundamental problem in Computer Science, no better algorithms are known in the general case for over 60 years and it is known to be notoriously hard to improve upon.
The best known algorithm for the special case where both the weights and the values of the items are small and bounded by $M$ and $V$ respectively, is a result by Pisinger~\cite{Pisinger99} who presents an algorithm with runtime $O(n M V)$.

Even for the subset sum problem, which is a more restricted special case of knapsack where the value of every item is equal to its weight, the best known algorithm beyond the textbook algorithm by Bellman~\cite{Bellman57} was also an algorithm by Pisinger~\cite{Pisinger99} which runs in time $O(n M)$ until significant recent progress by Bringmann~\cite{Bringmann17} and Koiliaris and Xu~\cite{KX17} was able to bring its the complexity down to $\widetilde{O}(n + T)$.

However, recent evidence shows that devising a more efficient algorithm for the general Knapsack problem is much harder. Specifically, \cite{CMWW17, KPS17}
reduce the $(\max,+)$-convolution problem to Knapsack, proving that any truly
subquadratic algorithm for Knapsack (i.e. $O((n + T)^{2-\epsilon})$) would imply a truly subquadratic algorithm for the $(\max,+)$-convolution problem. 
The problem of $(\max,+)$-convolution is a fundamental primitive inherently
embedded into a lot of problems and has been used as evidence for hardness for various problems in the last few years (e.g. \cite{CMWW17, KPS17, BIS17}).
However, an important open question remains here: \emph{Can we get faster algorithms that circumvent this conditional lower bound?}

We answer this question affirmatively by providing an algorithm running in time $O(TD)$, where $D$ is the number of \emph{distinct} weights. 
Our algorithm is deterministic and computes the optimal Knapsack value for \emph{all} capacities $t$ from $1$ to $T$.
Since $D\leq n$, its runtime either matches (for $D=\Theta (n)$) or yields an improvement (for $D=o(n)$) over Bellman's algorithm \cite{Bellman57}, for all parameter regimes.
It also directly implies runtimes of 
$O(TM)$\footnote{Concurrent and independent work by Bateni, Hajiaghayi, Seddighin, and Stein \cite{BHSS18} also obtains an algorithm running in time $\widetilde{O}(TM)$, 
as well as an algorithm running in time $\widetilde{O}(TV)$.
In comparison to ours, their $\widetilde{O}(TM)$ algorithm is randomized and computes the answer only for a single capacity $T$.},
$O(nM^2)$\footnote{Eisenbrand and Weismantel \cite{EW18} 
develop fast algorithms for Integer Programming. Concurrently and independently, they also obtain an algorithm for Knapsack
that runs in time $O(nM^2)$. They provide a structural property of Knapsack using the Steinitz Lemma that enables us to remove logarithmic factors in $T$ from our results for Unbounded Knapsack
(Theorem~\ref{unbounded}), as they reduce to the case $T=\Theta(M^2)$. Combined with Theorem~\ref{knapsack}, this also implies an $O(M^3)$ algorithm for Knapsack.}, 
and $O(nV^2)$, and therefore also yields an improvement over the $O(nMV)$ algorithm of Pisinger \cite{Pisinger99}.

Our algorithm can be summarized as follows: First, it partitions the items into $D$ sets according to their weights and solves the knapsack problem in each set of the partition for every possible capacity up to T. This can be done efficiently in $O(T)$ time as all items in each set have the same weight and thus knapsack can be greedily solved in those instances. Having a sequence of solutions for every capacity level for each set of items allows us to obtain the overall solution by performing $(\max,+)$-convolutions among them. 
Even though it is not known whether computing general $(\max,+)$-convolutions in truly sub-quadratic time is possible,
we exploit the inherent concavity of the specific family of sequences produced by our algorithm to perform this in linear time. We present our results in Section~\ref{s:01knap}.

In addition to the general Knapsack problem studied above, we also consider the \emph{Unbounded Knapsack} problem where there are infinite copies of every item. In Section~\ref{s:Unknap}, we present novel algorithms for Unbounded Knapsack with running times $O(M^2)$\footnote{Jansen and Rohwedder \cite{JR18} extend the results of \cite{EW18} for Integer Programming
and also concurrently and independently obtain an algorithm for Unbounded Knapsack 
running in time $O(M^2)$.} and $O(V^2)$, where $M$ is the maximum weight and $V$ is the maximum value of any item. Our algorithm again utilizes $(\max,+)$-convolutions of short sequences to compute the answer and interestingly is only pseudo-polynomial with respect to the maximum weight $M$ or the maximum value $V$ and not the capacity $T$.

Our results are summarized in Table~\ref{tab:results}.

\begin{table*}[t]
\def\arraystretch{1.2}
\centering
\begin{tabular}{ m{6cm}l l}
    \hline
    \textbf{Setting} & \textbf{Our Results} & \textbf{Conditional Lower bounds} \\
    \hline\hline
    \textbf{Knapsack}\: &  &  \\
    No bounds on weights or values\: & $O(TD)$ \textbf{\small{[Theorem~\ref{knapsack}]}} & $\Omega((TD)^{1-o(1)})$ \cite{CMWW17,KPS17} \\
    Weights bounded by $M$ \: & $O(TM)$ \textbf{\small{[Corollary~\ref{tmcorol}]}} & $\Omega((TM)^{1-o(1)})$ \cite{CMWW17,KPS17} \\
    Values bounded by $V$ \: & $O(nV^2)$ \textbf{\small{[Corollary~\ref{vcorol}]}} & \ \ \ \ \ \ \ \ \textbf{--} \\
    \hline
    \textbf{Unbounded Knapsack}\: &  &  \\
    Weights bounded by $M$ \: & $O(M^2)$ \textbf{\small{[Corollary~\ref{munbound}]}}& $\Omega(M^{2-o(1)})$ \cite{CMWW17,KPS17} \\
    Values bounded by $V$ \: & $O(V^2)$ \textbf{\small{[Corollary~\ref{vunbound}]}}& \ \ \ \ \ \ \ \ \textbf{--} \\
    \hline
    \hline
\end{tabular}
\caption{\small Summary of our deterministic pseudopolynomial time results on the Knapsack problem with the corresponding known conditional lower bounds based on $(\min,+)$-convolution.}
\label{tab:results}
\end{table*}

It follows from the results of \cite{CMWW17, KPS17} that, under the $(\min,+)$-convolution hardness assumption, 
it is not possible to obtain faster runtimes for Knapsack under most of the parameterizations that we consider. 
This is because, even though the lower bound claimed in these results is $\Omega((n+T)^{2-o(1)})$, the hardness construction uses a Knapsack instance where $T$, $M$, and $D$ are $\Theta(n)$.

\newpar{Capacitated Dynamic Programming} In addition to our results on the knapsack problem, we move on to study capacitated problems in a more general setting.
Specifically, we consider the problem of computing a path of maximum reward between a pair of nodes in a weighted Directed Acyclic Graph, where
the capacity constraint corresponds to an upper bound on the length of the path.

This model has successfully been used for uncapacitated problems \cite{Wilber88, KPS17}, as well as capacitated problems with 
weighted adjacency matrices that satisfy a specific condition, namely the \emph{Monge property} \cite{AST94, Schieber98, BLP92}.
In \cite{BLP92}, it is shown that under this condition, the maximum weight of a path of length $k$ is \emph{concave} in $k$. Whenever such a concavity property is true, one can always solve the capacitated problem by replacing the capacity constraint with an ``equivalent'' cost per edge. This cost can be identified through a binary search procedure that checks whether the solution for the uncapacitated problem with this cost corresponds to a path of length $k$.

Our second main result, Theorem~\ref{characterization}, gives a complete characterization of such a concavity property for transitive node-weighted graphs. We show that this holds if and only if the following graph theoretic condition is satisfied:

\noindent \emph{For every path $a \rightarrow b \rightarrow c$ of length 2, and every node $v$, at least one of the edges $(a,v)$ and $(v,c)$ exists.}

To illustrate the power of our characterization, we show that a linear algorithm can be easily obtained for the problem of $k$-sparse $\Delta$-separated subsequences~\cite{HDC09} recovering recent results of \cite{BS17,MMS18}.

To complement our positive result which allows us to obtain fast algorithms for finding maximum weight paths of length $k$, we provide strong evidence of hardness for transitive node-weighted graphs which do not satisfy the conditions of our characterization. We base our hardness results on computational assumptions for the $(\max,+)$-convolution problem we described above.

Beyond node-weighted graphs, when there are weights on the edges, no non-trivial algorithms are known other than for Monge graphs. Even in that case, we show that linear time solutions exist only if one is interested in finding the max-weight path of length $k$ between only one pair of nodes. If one is interested in computing the solution in Monge graphs for a single source but all possible destinations, we provide an algorithm that computes this in near-linear time in the number of edges in the graph.

\section{Preliminaries}

We first describe the problems of Knapsack and Unbounded Knapsack:

\begin{definition}[Knapsack]
Given $N$ items with weights $w_1, \dots, w_N\in[M]$ and values $v_1, \dots, v_N\in[V]$, and a parameter $T$, our goal is to find a set of items $S\subseteq [N]$ of total weight at most $T$
(i.e. $\sum\limits_{i\in S} w_i \leq T$) that maximizes the total value $\sum\limits_{i\in S} v_i$.
We will denote the number of distinct weights by $D$.
\end{definition}

\begin{definition}[Unbounded Knapsack]
Given $N$ items with weights $w_1, \dots, w_N\in[M]$ and values $v_1, \dots, v_N\in[V]$, and a parameter $T$, our goal is to find a multiset of items $S\subseteq [N]$ of total weight at most $T$
(i.e. $\sum\limits_{i\in S} w_i \leq T$) that maximizes the total value $\sum\limits_{i\in S} v_i$.
We will denote the number of distinct weights by $D$.
\end{definition}

Throughout the paper we make use of the following operation between two sequences called $(\max,+)$-convolution.

\begin{definition}[$(\max,+)$-convolution]
Given two sequences $a_0, \dots, a_n$ and $b_0, \dots, b_m$, the $(\max,+)$-convolution $a\oplus b$ between $a$ and $b$ is a sequence $c_0, \dots, c_{n+m}$ such that for any $i$
\[ c_i = \underset{0\leq j\leq i}{\max}\left\{ a_j + b_{i-j} \right\} \]
This operation is commutative, so it is also true that
\[ c_i = \underset{0\leq j\leq i}{\max}\left\{ a_{i-j} + b_{j} \right\} \]
\end{definition}

Our algorithms rely on uncovering and exploiting discrete concavity that is inherent in the problems we consider.

\begin{definition}[Concave, $k$-step concave]
A sequence $b_0, \dots, b_{n}$ is \emph{concave} if for all $i\in\{1,\dots,n-1\}$ we have $b_i - b_{i-1} \geq b_{i+1} - b_i$. A sequence is called
\emph{$k$-step concave} if its subsequence $b_0, b_k, b_{2k}, \dots $ is concave and for all $i$ such that $i\mod k \neq 0$, we have that
$b_i = b_{i-1}$.
\end{definition}

For the problems defined on graphs with edge weights, we typically assume that their weighted adjacency matrix is given by a Monge matrix.

\begin{definition}[Monge matrices]
A matrix $A\in\mathbb{R}^{n\times m}$ is called Monge if for any $1 \leq i \leq n-1$ and $1 \leq j \leq m-1$
\[ A_{i,j} + A_{i+1,j+1} \geq A_{i+1,j} + A_{i,j+1} \]
\end{definition}

\begin{definition}[Monge weights]
We will say that a Directed Acyclic Graph has \emph{Monge weights} if its adjacency matrix is a Monge matrix.
\end{definition}

In addition to our positive results, we present evidence of computational hardness assuming for
$(\max,+)$-convolution problem.

\begin{definition}[$(\max,+)$-convolution hardness]
The $(\max,+)$-convolution hardness hypothesis states that any algorithm that computes the $(\max,+)$-convolution of two sequences of size $n$
requires time $\Omega(n^{2-o(1)})$.
\end{definition}

A result of \cite{BIS17} shows that the $(\max,+)$-convolution problem is equivalent to the following problem:
Given an integer 
$n$ and three sequences $a_0, \dots, a_n$, $b_0, \dots, b_n$, and $c_0, \dots, c_n$, compute $\underset{i+j+k=n}{\max}\{a_{i}+b_{j}+c_{k}\}$. In our conditional lower bounds, we will be using this equivalent form of the conjecture.

\section{Knapsack}
In this section we present two novel pseudo-polynomial deterministic algorithms, one for Knapsack and one for Unbounded Knapsack. The running times of these algorithms significantly improve
upon the best known running times in the small-weight regime.
In essence, the main improvements stem from a more principled understanding and systematic use of $(\max,+)$-convolutions.
Thus, we show that devising faster algorithms for special cases of $(\max,+)$-convolution lies in the core of improving algorithms for the Knapsack problem.
In Theorem~\ref{knapsack}, we present an algorithm for Knapsack that runs in time $O(TD)$, where $T$ is the size of the knapsack and $D$ is the number of distinct item weights.
Then, in Theorem~\ref{unbounded}, we present algorithms for Unbounded Knapsack with runtimes $O(M^2)$ and $O(V^2)$, where $M$ is the maximum weight and $V$ the maximum value of some item.

\subsection{Knapsack}
\label{s:01knap}
Given $N$ items with weights $w_1, \dots, w_N\in[M]$ and values $v_1, \dots, v_N\in[V]$, and a parameter $T$, 
our goal is to find a set of items $S\subseteq [N]$ of total weight at most $T$
(i.e. $\sum\limits_{i\in S} w_i \leq T$) that maximizes the total value $\sum\limits_{i\in S} v_i$.
We will denote the number of distinct weights by $D$.

\begin{algorithm}[ht]
\caption{Knapsack}
\begin{algorithmic}[1]
\STATE Given items with weights in $\{w_1^{\#}, \dots, w_D^{\#}\}$
\STATE Partition items into sets $S_1, \dots, S_D$, so that $S_i = \{j\ |\ w_j = w_i^{\#}\}$
\FOR{$i\in[D]$ and $t\in[T]$}
\STATE $b^{(i)}_t\leftarrow$ solution for $S_i$ with knapsack size $t$
\ENDFOR
\STATE $s \leftarrow$ empty sequence
\FOR{$i\in[D]$}
\STATE $s\leftarrow s\oplus b^{(i)}$ using Lemma~\ref{stepconv}
\STATE Truncate $s$ after the $T$-th entry
\ENDFOR
\STATE Output $s_T$
\end{algorithmic}
\label{algo1}
\end{algorithm}

The following is the main theorem of this section.
\begin{theorem}
Algorithm~\ref{algo1} solves \emph{Knapsack} in time $O(T D)$.
\label{knapsack}
\end{theorem}

\newpar{Overview} The main ingredient of this result is an algorithm for fast $(\max,+)$-convolution in the case that one of the two sequences is \emph{$k$-step concave}.
Using the SMAWK algorithm \cite{SMAWK} it is not hard to see how to do this in linear time for $k=1$. 
For the general case, we show that computing the $(\max,+)$-convolution of the two sequences can be decomposed into 
$\frac{n}{k}$ subproblems of computing the $(\max,+)$-convolution between two size-$k$ subsequences of the two sequences. Furthermore, 
the subsequence that came from the $k$-step concave sequence is concave and so each subproblem can be solved in time $O(k)$ and the total time spent in the subproblems will be $O(\frac{n}{k} k) = O(n)$.

\begin{lemma}
Given an arbitrary sequence $a_0, \dots, a_m$ and a concave sequence $b_0,\dots, b_n$ we can compute the $(\max,+)$ convolution between $a$ and $b$ in time $O(m + n)$.
\label{conv}
\end{lemma}
\begin{proof}
Consider the matrix $A$ with
$A_{ij} = a_j + b_{i - j}$, where we suppose that elements of the sequences with out-of-bounds indices have value $-\infty$.
Note now that $(a\oplus b)_i$ is by definition equal to the maximum value of row $i$ of $A$. Therefore computing $a\oplus b$ corresponds to finding the row maxima of $A$.
Now note that for any $(i,j)\in\{0,1,\dots,n-1\}\times\{0,1,\dots,m-1\}$,
we have 
\begin{align*} A_{i,j} - A_{i,j+1} = & a_j + b_{i-j} - a_{j+1} - b_{i-j-1} \\
\overset{concavity}{\geq}& a_j + b_{i+1-j} - a_{j+1} - b_{i-j} \\
 =& A_{i+1,j} - A_{i+1,j+1} \end{align*}
therefore
$A$ is Monge. 
The main result of \cite{SMAWK} is that given a Monge matrix $A\in\mathbb{R}^{n\times m}$, one can compute
all its row maxima in time $O(m + n)$, which implies the Lemma. 
\end{proof}

\begin{lemma}
Given an arbitrary sequence $a_0, \dots, a_m$ and a $k$-step concave sequence $b_0,\dots, b_n$ we can compute the $(\max,+)$ convolution of $a$ and $b$ in time $O(m + n)$.
\label{stepconv}
\end{lemma}
\begin{proof}
We use the fact that we can compute the $(\max,+)$ convolutions of an arbitrary sequence with a concave sequence in linear time (Lemma~\ref{conv}). Since $b$ is a $k$-step concave sequence, taking
every $k$-th term of it one gets a concave sequence of size $O(n/k)$. Then, we do the same for $a$, taking $k$ subsequences of size $m/k$ each. Therefore we can compute the convolution
between the concave sequence and all of these subsequences of $a$ in linear time. The results of these convolutions can be used to compute the final sequence.
We now describe this in detail.

For ease of notation, we will again assume that our sequences take value $-\infty$ in out-of-bounds indices.
Let $x^{(i)} := (a_{i}, a_{k + i}, a_{2k + i}, \dots)$ denote the subsequence of $a$ with indices whose remainder is $i$ when divided by $k$,
and $y := (b_0, b_k, b_{2k}, \dots)$.
Furthermore, define 
\[f_i = \displaystyle \max_{q=0}^{\infty} \{ b_{qk} + a_{i - qk} \}\]
Now, for any $j$ we have 
\begin{align*}
 \displaystyle\max_{i=j-k+1}^{j} f_i &= \max_{i=j-k+1}^{j} \max_{q=0}^{\infty} \{ b_{qk} + a_{i - qk} \} \\
 &=\max_{i=j-k+1}^{j} \max_{q=0}^{\infty} \{ b_{qk + j - i} + a_{i - qk} \} \\
 &= \max_{z=0}^{\infty} \{ b_z + a_{j - z} \} 
\end{align*}

where the second equality follows from the fact that $b_{qk + t} = b_{qk}$ for any $t\in[k-1]$ and the third from the fact that $z=qk+j-i$ can take any value in $[0,\infty)$.

This is the $j$-th element of the $(\max,+)$-convolution between $a$ and $b$, so the elements of this convolution are exactly the maxima of size-$k$ segments of $f$.

In order to compute $f$, note that for some $p$, the convolution between $x^{(p)}$ and $y$ gives us all values of $f$ of the form $f_{qk+p}$, for any $q$.
This is because from the definition of $f$,
\begin{align*} f_{qk + p} &= \displaystyle \max_{z=0}^{\infty} \{ b_{zk} + a_{qk + p - zk} \} \\
  &= \max_{z=0}^{\infty} \{ y_z + x^{(p)}_{q - z}\} \\
  &= (x^{(p)}\oplus y)_q 
\end{align*}
Furthermore, $y$ is a concave sequence and by Lemma~\ref{conv} we can compute such a convolution in time $O((m + n)/k)$. Doing this for all $0\leq p < k$, we can compute all
values of $f$ in time $O(m + n)$.

Now, in order to compute the target sequence, we have to compute the maxima of all size-$k$ segments of $f$. We can do that using a simple sliding window technique.
Specifically, suppose that for some segment $[i,i+k-1]$ we have an increasing subsequence of $f_{[i,\dots,i+k-1]}$, containing all the potentially useful elements.
The first element of this subsequence is the maximum value of $f$ in the segment $[i,i+k-1]$. Now, to move to $[i+1,i+k]$, we remove $f_i$ if it is in the subsequence,
and then we compare $f_{i+k}$ with the last element in the subsequence. Note that if that last element has value $\leq f_{i+k}$, it will never be the maximum element 
in any segment. Therefore we can remove it and repeat until the last element has value greater than $f_{i+k}$, at which point we just insert $f_{i+k}$ in the end of the subsequence.
Note that by construction, this subsequence will always be decreasing, and the first element will be the maximum of the respective segment.
The total runtime is linear if implemented with a standard queue.
\end{proof}

Now that we have these tools we can use them to prove the main result of this section:

\medskip
\begin{prevproof}{Theorem}{knapsack}
Consider any knapsack instance where $D$ is the number of distinct item weights $w_1^{\#}, \dots, w_D^{\#}$. Now for each $i\in[D]$ let $c_i$ be the number of items with weight $w_i^{\#}$ and
$v_1^{(i)} \geq v_2^{(i)} \dots \geq v_{c_i}^{(i)}$ their respective values. 

If we only consider items with weights $w_i^{\#}$, the knapsack problem is easy to solve, since we will just greedily 
pick the most valuable items until the knapsack fills up. More specifically, if $b_s$ is the maximum value obtainable with a knapsack of size $s$, we have that
$b_0=0, b_{w_i} = v_1^{(i)}, b_{2w_{i}} = v_1^{(i)} + v_2^{(i)}, \dots$, and also $b_j = b_{j-1}$ for any $j$ not divisible by $w_i^{\#}$.
Therefore $b$ is a $w_i^{\#}$-step concave sequence.

In order to compute the full solution, we have to compute the $(\max,+)$ convolution of $D$ such sequences. Since by Lemma~\ref{stepconv} each convolution takes linear time
and we only care about the first $T$ values of the resulting sequence (i.e. we will only ever keep the first $T$ values of the result of a convolution),
the total runtime is $O(TD)$, where $T$ is the size of the knapsack.
\end{prevproof}

\begin{corollary}
\label{tmcorol}
Knapsack can be solved in $O(T M)$ time.
\end{corollary}
\begin{corollary}
\label{vcorol}
Knapsack can be solved in time $O(nM^2)$ or $O(nV^2)$.
\end{corollary}
\begin{proof}
The first bound directly follows by Corollary~\ref{tmcorol} and the fact that $T \leq nM$.
For the second bound, note that by swapping the role of the weights and the values in Algorithm~\ref{algo1}, replacing all $(\max,+)$-convolutions by $(\min,+)$-convolutions,
and setting the knapsack capacity to $nV$ as opposed to $T$, this algorithm runs in time $O(nV^2)$
and outputs for every possible value, the minimum weight of items that can achieve this value. The answer can then be recovered by finding the minimum 
value that gives a corresponding weight of at most $T$.
\end{proof}

\subsection{Unbounded Knapsack}
\label{s:Unknap}
Given $N$ items with weights $w_1, \dots, w_N$ and values $v_1, \dots, v_N$, and a parameter $T$, 
our goal is to find a multiset of items $S\subseteq [N]$ of total weight at most $T$
(i.e. $\sum\limits_{i\in S} w_i \leq T$) that maximizes the total value $\sum\limits_{i\in S} v_i$.
We will denote the largest item weight by $M$.

Note that this problem is identical to Knapsack except for the fact that there is no
limit on the number of times each item can be picked. This means that
we can assume that there are no two items with the same weight, since we would only ever pick the most valuable of the two.

\begin{algorithm}[ht]
\caption{Unbounded Knapsack}
\begin{algorithmic}[1]
\STATE Let $v^{(0)}$ be a sequence where $v^{(0)}_x$ is the value of the element with weight $x$ or $-\infty$ if no such element exists
\FOR{$z=1,\dots,\lceil\log M\rceil$}
\STATE $v^{(z)} \leftarrow v^{(z-1)}\oplus v^{(z-1)}$
\ENDFOR
\STATE $a_{\left[0,M\right]} \leftarrow v^{(\lceil\log M\rceil)}$
\FOR{$i=\lceil\log\frac{T}{M}\rceil,\dots,1$}
\STATE $a_{\left[\frac{T}{2^i} - M, \frac{T}{2^i} + M\right]} \leftarrow a_{\left[\frac{T}{2^i} - M, \frac{T}{2^i}\right]} \oplus a_{\left[0,M\right]}$
\STATE $a_{\left[\frac{T}{2^{i-1}}-M, \frac{T}{2^{i-1}} \right]} \leftarrow a_{\left[\frac{T}{2^{i}} - M, \frac{T}{2^i} + M\right]} \oplus a_{\left[\frac{T}{2^{i}} - M, \frac{T}{2^i} + M\right]}$
\ENDFOR
\STATE Output $a_T$
\end{algorithmic}
\label{algo2}
\end{algorithm}

The following is the main theorem of this section:
\begin{theorem}
Algorithm~\ref{algo2} solves \emph{Unbounded knapsack} in time $O(M^2 \log T)$.
\label{unbounded}
\end{theorem}

\newpar{Overview} As in the algorithm for Knapsack our algorithm utilizes $(\max,+)$-convolutions, but with a different strategy. 
We aren't using any concavity arguments here, but in fact we will use the straightforward quadratic-time algorithm for computing $(\max,+)$-convolutions.
The main argument here is that if all the weights are relatively small, one can always partition any solution in two, so that the weights of the two parts 
are relatively close to each other. Therefore, for any knapsack size we only have to compute the optimal values for a few knapsack sizes around its half, and not for
all possible knapsack sizes.

We can now proceed to the proof of this result.

\medskip
\begin{prevproof}{Theorem}{unbounded}
Consider any valid solution to the unbounded knapsack instance. Since every item has weight at most $M$, we can partition the items of that solution into two multisets with
respective weights $W_1$ and $W_2$, so that $|W_1 - W_2| < M$ (one can obtain this by repeatedly moving any item from the larger part to the smaller one).
This implies the following, which is the main fact used in our algorithm: If $a_s$ is the maximum value obtainable with a knapsack of size $s$, then we have that 
\[ a_s = \left[(a_{s/2 - M/2}, \dots, a_{s/2 + M/2})^{\oplus 2}\right]_s \]
where ${}^{\oplus 2}$ denotes $(\max,+)$-convolution squaring, i.e. applying $(\max,+)$-convolution between a sequence and itself.

First, we compute the values $a_1, \dots, a_M$ in $O(M^2\log M)$ time as follows:
We start with the sequence $v^{(0)}$, where $v_x^{(0)}$ is the value of the element with weight $x$, or $-\infty$ if such an item does not exist.
Now define $v^{(i+1)} = (v^{(i)} \oplus v^{(i)})_{[0,M]}$. This convolution can be applied in time $O(M^2)$ for any $i$, since we are always only keeping the first $M$ entries.
By induction, it is immediate that $v^{(i)}$ contains the optimal values achievable for all 
knapsack sizes in $[M]$ using at most $2^i$ items. Therefore $a_{0,\dots,M} \equiv v_{0,\dots,M}^{(\lceil\log M\rceil)}$,
which as we argued can be computed in time $O(M^2 \log M)$.

Now, suppose that we have computed the values $a_{\frac{T}{2^i} - M}, \dots, a_{\frac{T}{2^i}}$ for some $i$.
By convolving this sequence with $a_0, \dots, a_M$ we can compute in time $O(M^2)$ the values $a_{\frac{T}{2^i}+1}, \dots, a_{\frac{T}{2^i}+M}$.
Now, convolving the sequence
$a_{\frac{T}{2^i}-M}, \dots, a_{\frac{T}{2^i}+M}$
with itself gives us $a_{\frac{T}{2^{i-1}} - M}, \dots, a_{\frac{T}{2^{i-1}}}$ 
(here we used the fact that to compute $a_{2j}$ we only need the values $a_{j - M/2},\dots,a_{j + M/2}$). 
Doing this for $i=\lceil\log T\rceil, \dots, 2, 1$, we are able to compute the values $a_{T-M}, \dots, a_T$ in total time $O(M^2 \log T)$. 
The answer to the problem, i.e. the maximum value achievable, is $\max\{a_{T-M}, \dots, a_T\}$.
\end{prevproof}

Recent work of \cite{EW18} shows, using the \emph{Steinitz lemma}, that an optimal Knapsack solution for a capacity in $[T-M,T]$ can be turned into an optimal solution for capacity $T$
by inserting or removing at most $M$ elements, where $M$ is a bound on weight of the items.
In the case of Unbounded Knapsack, a solution that only uses the best item until it exceeds capacity $T-M^2$ can always be extended into an optimal solution with capacity $T$.
Therefore the capacity can be assumed to be $O(M^2)$. Combining this with the $\frac{M^2}{2^{\sqrt{\log M}}}$-time algorithm of \cite{Williams14} for $(max,+)$-convolution implies an
algorithm that runs in time $O\left(\frac{M^2}{2^{\sqrt{\log M}}}\right)$.
\begin{corollary}
\label{munbound}
Unbounded Knapsack can be solved in time $O(M^2)$
\end{corollary}
A similar argument can be used to get a more efficient algorithm when we have a bound on the values of the items. In particular, using the item $j$ with the highest value-to-weight ratio
$v_j / w_j$, $k = \lfloor \frac{T}{w_j} \rfloor + 1$ times, until we exceed the capacity we get both a lower bound of $(k-1) v_j$ and a upper bound of $k v_j$ on the value of the optimal
solution. In addition, again by the Steinitz Lemma, there exists an optimal solution that uses item $j$ at least $k - V$ times. This allows us to start from value $(k - V) v_j$ and
compute the minimum weight required to achieve values in $[(k-1)v_j, kv_j]$. This gives an algorithm that runs in $O(V^2 \log V)$ using the naive algorithm for $(\min,+)$-convolutions,
and $O\left(\frac{V^2}{2^{\sqrt{\log V}}}\right)$ using the improved algorithm by $\cite{Williams14}$.
\begin{corollary}
\label{vunbound}
Unbounded Knapsack can be solved in time $O(V^2)$
\end{corollary}

\section{Capacitated Dynamic Programming}
We now move on to study more general capacitated dynamic programming settings, described by computing the maximum reward path of length $k$ in a directed acyclic graph.
This setting can capture a lot of natural problems, either directly or indirectly. In the following theorem, we show that the Knapsack problem is a special case of this model and
thus a better understanding of Knapsack can lead to improved algorithms for other capacitated problems.

\begin{lemma}[Knapsack]
The Knapsack problem can be modeled as finding a maximum reward path with at most $k$ edges in a node-weighted transitive DAG.
\end{lemma}
\begin{proof}
Let the item weights and values be $w_1, \dots, w_n$ and $v_1, \dots, v_n$ respectively. We will create a DAG for each item and then join all these DAGs in series.
Specifically, for item $i$, its DAG $G_i$ will consist of two parallel paths $Y_i$ and $N_i$ between a pair of nodes $s_i$ and $t_i$.
$Y_i$ will correspond to taking item $i$, and $N_i$ to not taking it. 
\begin{itemize}
\item{$N_i$ will be a path of length $2$ from $s_i$ to $t_i$, where the intermediate vertex has reward $b = nw_{max}v_{max}$.}
\item{$Y_i$ will be a path of length $w_i+2$ from $s_i$ to $t_i$, all of which vertices other than $s_i, t_i$ have reward $\frac{b + v_i}{w_i+1}$.}
\end{itemize}
Finally, we just join all $G_i$ in series, i.e. for all $i\in[n-1]$, identify $t_i$ with $s_{i+1}$, and we ask for the maximum reward path with at most $n + T$ edges from $s_1$ to $t_n$.

Note that we couldn't have just set the reward of path $N_i$ to $0$, because that would potentially allow one to pick a path that is a subset of $Y_i$, which corresponds to picking a fraction
of an item and is invalid.
However, note that with our current construction any optimal path will either use the whole path $Y_i$, or it will not use it at all.
This is because 
the reward of picking $1$ vertex from $N_i$ is $b$, while the reward of picking
at most $w_i$ vertices from $Y_i$
is at most $b + v_i - \frac{b + v_i}{w_i+1} = b - \frac{b - w_i v_i}{w_i+1} = b - \frac{nw_{max}v_{max} - w_i v_i}{w_i+1} < b$. 

Now, for each $i$, the optimal path will definitely contain either $Y_i$ or $N_i$. If this were not the case, the reward of the solution would be at most
\begin{align*}
\sum\limits_i \left(b + v_i\right) - \underset{i}{\min} \{b + v_i\} < nb + \sum\limits_i v_i - b \leq nb + nv_{max} - nw_{max}v_{max} < nb
\end{align*}
while
by simply picking all $N_i$'s one gets reward $nb$.

Therefore we have shown that the total reward of the optimal solution will be 
\[ \sum\limits_{i\in S} (b + v_i) + b (n - \left|S\right|) = nb + \sum\limits_{i\in S} v_i \]
for some set $S$ such that 
$n + T\geq \sum\limits_{i\in S} (w_i+1) + (n - |S|) = n + \sum\limits_{i\in S} w_i$, or equivalently $\sum\limits_{i\in S} w_i \leq T$.
Therefore this set $S$ is the optimal set of items to be picked in the knapsack.

\end{proof}

	\subsection{Node-weighted Graphs} \label{s:node}
In this section, we study the problem of finding maximum-reward paths in node-weighted transitive DAGs. 
In Lemma~\ref{transitive}, we show that in general this problem is hard, by reducing $(\max,+)$-convolution to it.
We then proceed to show our second main result, which provides a family of graphs for which the problem can be efficiently solved.


\begin{lemma}[$(\max,+)$-hardness of Node-weighted graphs]
Given a transitive DAG, a pair of vertices $s$ and $t$, and an integer $k$, the problem of computing a maximum reward path from $s$ to $t$ with at most $k$ edges
is $(\max,+)$-convolution hard, i.e. requires $\Omega((m k)^{1-o(1)})$ time assuming $(\max,+)$-convolution hardness.
\label{transitive}
\end{lemma}
\begin{proof}
Given a sequence $x_0, \dots, x_k$, we construct the following node-weighted DAG, on nodes $a_0, \dots, a_k, a_0', \dots, a_k'$.
For all $i\in[k]$, we add edge $(a_{i-1}, a_i)$ and
for all $0\leq i \leq k$, we add edge $(a_i, a_i')$.
Let $M = 3 \max\{|x_0|, \dots, |x_k|\}$ and $T = 10 M k$.
If we denote the value of node $z$ as $val(z)$, we define $val(a_i)=M$ and $val(a_i') = T + x_i - M i$ for all $i$.

Now consider three such DAGs, one for each subsequence $x_0, \dots, x_k$, $y_0, \dots, y_k$, and $z_0, \dots, z_k$, with the node sets being
$a_{\star}$ and $a_{\star}'$, $b_{\star}$ and $b_{\star}'$, and $c_{\star}$ and $c_{\star}'$ respectively.
We connect them in series, i.e.
each node of the first DAG has an edge to $b_0$, and each node of the second DAG has an edge to $c_0$. Then we take the transitive closure of the resulting DAG.

Note that any maximum reward path with $k+5$ hops on this DAG will necessarily use some $a_{\star}'$, $b_{\star}'$, and $c_{\star}'$.
If this were not the case, the value to be obtained would be less than 
\[M(k+4) + 2(T + M) \leq 2.5 T\]
However, a path containing some $a_i', b_j', c_l'$ for some $i,j,l$ will have weight at least
\[3T - 3Mk \geq 2.5 T\]

Furthermore, the first nodes of the path will be of the form
$a_0, a_1, \dots, a_i, a_i'$. To see this, suppose otherwise, i.e. that the path contains $j + 1 < i + 1$ nodes of the form $a_{\star}$. Then the total value of the 
part of the path up to $a_i'$ will be
\begin{align*}
& M (j + 1) + T + x_i - M i 
 \leq M(j + 1) + T + x_i - M - M j \\
& < M(j+1) + T + x_{j} - M j
 = a_0 + a_1 + \dots + a_j + a_j'
 \end{align*}
so the path $a_0,a_1, \dots, a_{j}, a_{j}'$ is always better and has the same number of edges.
A similar argument can be applied to the rest of the path, to show that the total maximum length $(k+5)$-hop path will be of the form
$a_0, \dots, a_{k_1}, a_{k_1}', b_0, \dots, b_{k_2}, b_{k_2}', c_0, \dots, c_{k_3}, c_{k_3}'$, with $k_1 + k_2 + k_3 = k$.
The only remaining case is that of the path containing some edge $(a_{k_1}',b_j')$, for some $j$ (and similarly for $(b_{k_2}',c_j)$). However in this case we can find a better path.
Suppose that $k_1 \geq 1$ (otherwise we can do it for $k_3$). Replace edges $(a_0,a_1)$ and $(a_1,u)$ of the path with edge $(a_0, u)$, essentially skipping over $a_1$,
and also replace edge $(a_{k_1}',b_j')$ by edges $(a_{k_1}',b_0)$ and $(b_0, b_j')$. Note that both the value and the length remained the same, but we use less than $k_1+1$ $a_{\star}$ nodes,
so this path is not optimal, as seen by the argument we stated before.

In light of the above, a maximum-weight $(k+5)$-hop path in this graph will be of the form
\[ a_0, \dots, a_{k_1}, a_{k_1}',b_0, \dots, b_{k_2}, b_{k_2}', c_0, \dots, c_{k_3}, c_{k_3}'\]
where $k_1+k_2+k_3=k$, and have value equal to
$3T + x_{k_1} + y_{k_2} + z_{k_3}$. Therefore computing $\max_{k_1+k_2+k_3=k} \{x_{k_1} + y_{k_2} + z_{k_3}\}$ is equivalent to finding the maximum value $(k+5)$-hop path on this DAG.
Note that if given in succinct form the number of edges is linear in $k$, so
any $O((mk)^{1-\epsilon})$ algorithm for this problem
yields
an $O(k^{2-\epsilon})$ algorithm for $(\max,+)$-convolution.
\end{proof}

As we saw in the introduction, one can solve the problem if the optimal value as a function of the capacity is concave. This is made formal in the following lemma:
\begin{lemma}[Concave functions]
Let $G$ be a node-weighted transitive DAG with $n$ vertices and $m$ edges, whose weights' absolute values are bounded by $M$, and 
let $f(x)$ be the maximum reward obtainable in a path of length $x$.
If $f$ is a concave function, then one can reduce the capacitated
problem (i.e. computing $f(k)$ for some $k$) to solving $O(\log(nM))$ uncapacitated problems with some fixed extra cost per item.
Since each one of these problems can be solved in $O(m)$ time, the total runtime is $O(m\log(nM))$.
\label{concavity}
\end{lemma}

In Lemma~\ref{characterization} we give a complete graph-theoretic characterization of the graphs that have this concavity property and therefore can be solved efficiently.

\begin{lemma}[Concavity characterization]
The problem of finding a maximum reward path with at most $k$ edges in a transitive DAG
is concave for all choices of node weights if and only if for any path $u_1\rightarrow u_2 \rightarrow u_3$ and any node $v$ either $u_1\rightarrow v$ or $v\rightarrow u_3$ (Property $\mathcal{P}$).
\label{characterization}
\end{lemma}
\begin{proof}
Let $f(k)$ be the maximum reward obtainable with a path of exactly $k$ edges.

\noindent$\Rightarrow$: Let $G$ be a DAG for which property $\mathcal{P}$ doesn't hold. Let $u_1\rightarrow u_2\rightarrow u_3$ be the path of length $2$ and $v$ be the vertex that has no edge
to or from any of $u_1, u_2, u_3$. We set the node values as $val(u_1) = val(u_2) = val(u_3) = 1$, $val(v) = 1 + \epsilon$, and $-\infty$ for all other vertices. Then,
$f(1) = 1 + \epsilon$, $f(3) = 3$, but $f(2) = 2 < \frac{f(1) + f(3)}{2}$, therefore $f$ is not concave.

\noindent$\Leftarrow$: Suppose that property $\mathcal{P}$ is true. 
Now, let $P = (s, p_1, p_2, \dots, p_{k-1}, t)$ 
be a path of length $k$ such that $val(P) = f(k)$ and 
$Q = (s, q_1, q_2, \dots, q_{k+1}, t)$ be a path of length $k+2$ such that $val(Q) = f(k+2)$,
where $P$ and $Q$ can potentially have common vertices other than $s$ and $t$.
Since property $\mathcal{P}$ is true, we know that for any $i\in[k-1]$, there is either an edge 
from one of $q_i, q_{i+1}, q_{i+2}$ to $p_i$, or
from $p_i$ to one of $q_i, q_{i+1}, q_{i+2}$. 
By transitivity, this implies that either  
$q_i \rightarrow p_i$, or
$p_i\rightarrow q_{i+2}$.
We distinguish three cases. In all three cases we will be able to find paths $P'$ and $Q'$ with $k+1$ edges each, that contain all vertices of the form $p_i$ and $q_i$.

\newpar{Case 1: $q_1\rightarrow p_1$} \\
We pick $P' = (s, q_1, p_1, \dots, p_{k-1},t)$ and $Q' = (s, q_2, \dots, q_{k+1}, t)$.

\newpar{Case 2: $\forall i: p_i \rightarrow q_{i+2}$} \\
We pick $P' = (s, p_1, \dots, p_{k-1}, q_{k+1}, t)$ and $Q' = (s, q_1, \dots, q_k, t)$

\newpar{Case 3: $\exists i: p_i \rightarrow q_{i+2}$, $q_{i+1}\rightarrow p_{i+1}$} \\
We pick 
$P' = (s, p_1, \dots, p_i, q_{i+2}, \dots, q_{k+1}, t)$ and $Q' = (s, q_1, \dots, q_{i+1}, p_{i+1}, \dots, p_{k-1}, t)$.

\medskip
\noindent Therefore we established that in any case there exist such paths $P'$ and $Q'$. Now note that
\begin{align*}
\max\left\{val(P'), val(Q') \right\} \geq \frac{1}{2} \left(val(P') + val(Q')\right) = \frac{1}{2}\left(val(P) + val(Q)\right)
\end{align*}
and therefore $f$ is a concave function.
\end{proof}

As mentioned before, even very simple special cases of the model capture a lot of important problems.
In the following lemma, we show that we can solve the $k$-sparse $\Delta$-separated subsequence problem~\cite{HDC09}
in near-linear time using the main result of this section, thus recovering recent results of \cite{BS17,MMS18}.

\begin{lemma}[Maximum-weight $k$-sparse $\Delta$-separated subsequence]
Given a sequence $a_1, \dots, a_n$, find indices $i_1, i_2, \dots, i_k$ such that for all $j\in[k-1]$, $i_{j+1} \geq i_j + \Delta$ 
and the sum $\sum\limits_{j\in[k]} a_{i_j}$ is maximized.
This problem can be solved in $O(n\log\left(n\max_i \left|a_i\right|\right))$ time.
\end{lemma}
\begin{proof}
Let's define a simple node-weighted DAG for this problem. 
We define a sequence of vertices $u_1, \dots, u_n$ each one of which corresponds to picking an element from the sequence.
Then, we add an edge $u_i\rightarrow u_j$ iff $j - i \geq \Delta$.
Furthermore, for all $i$, $val(u_i) = a_i$.
It remains to prove that it satisfies the property of Lemma~\ref{characterization}.
Consider any length-$2$ path $u_i\rightarrow u_j\rightarrow u_k$. We know that both $k-j$ and $j-i$ are at least $\Delta$.
Now, for any $u_p$ we have that
\begin{align*}
\max\{\left|u_p - u_k\right|, \left|u_p - u_i\right|\} 
\geq \frac{1}{2}\left(\left|u_p - u_k\right| + \left|u_p-u_i\right|\right)
\geq \frac{1}{2}\left(\left|u_k - u_i\right|\right)
\geq \frac{1}{2} 2 \Delta
= \Delta 
\end{align*}
so there is an edge between $u_p$ and either $u_i$ or $u_k$. Therefore by Lemma~\ref{concavity}
the problem can be solved in time $O(m\log(n\max_i |a_i|)) = O(n^2 \log(n\max_i |a_i|))$.

The quadratic runtime stems from the fact that the DAG we constructed is dense. In fact, we can do better by defining some auxiliary vertices $v_1, \dots, v_n$.
The values of these extra vertices will be set to $-\infty$ to ensure that they aren't used in any solution and thus don't break the concavity.
Instead of edges between vertices $u_i$, we only add the following edges
\begin{itemize}
\item{$u_i\rightarrow v_i$ for all $i$}
\item{$v_i\rightarrow u_{i+\Delta}$ for all $i + \Delta \leq n$}
\item{$v_i\rightarrow v_{i+1}$ for all $i + 1 \leq n$}
\end{itemize}
Now, the number of edges is $O(n)$ and so the runtime becomes $O(n\log(n\max_i |a_i|))$.
\end{proof}

As another example of a problem that can be modeled as a capacitated maximum-reward path problem in a DAG, we consider the Max-Weight Increasing Subsequence of length $k$ problem.
In contrast to its uncapacitated counterpart, which is solvable in linear time, the capacitated version requires quadratic time, assuming $(\max,+)$-convolution hardness, as witnessed 
in the following lemma.

\begin{lemma}[Max-Weight Increasing Subsequence of length $k$]
Given a sequence $a_1, \dots, a_n$ with respective weights $w_1, \dots, w_n$, 
find indices $i_1 < i_2 < \dots < i_k$ such that for all $j\in[k-1]$, $a_{i_j} \leq a_{i_{j+1}}$ 
and the sum $\sum\limits_{j\in[k]} w_{i_j}$ is maximized.
This problem is $(\max,+)$-convolution hard, i.e. requires $\Omega((nk)^{1-o(1)})$ time assuming $(\max,+)$-convolution hardness. 
\end{lemma}
\begin{proof}
Consider the construction used in Lemma~\ref{transitive}. 
We define an instance of the Max-Weight Increasing Subsequence of length $k$
problem which contains an element for each node of the DAG. Specifically, let's define our sequence to be 
\begin{align*}
x_0, x_0', x_1, x_1',\dots, x_k, x_k', y_0, y_0', \dots, y_k, y_k', z_0, z_0',\dots, z_k, z_k'
\end{align*}
\begin{align*}
\text{with}\ \ \ 
&x_i = i &x_i' = 2k+1 - i\\
&y_i = 2k+2 + i &y_i' = 4k+3 - i\\
&z_i = 4k+4 + i &z_i' = 6k+5 - i
\end{align*}
where the 
weight of each element is equal to the weight of the corresponding node in the DAG (i.e. 
$x_\star\leftrightarrow a_\star$,
$x_\star'\leftrightarrow a_\star'$,
$y_\star\leftrightarrow b_\star$,
$y_\star'\leftrightarrow b_\star'$,
$z_\star\leftrightarrow c_\star$,
$z_\star'\leftrightarrow c_\star'$)

By definition of the sequence, the fact that we are looking for increasing subsequences
implies that there is a $1-1$ correspondence between length-$k$ increasing subsequences and $(k-1)$-hop paths of the original DAG.
Therefore any $O((nk)^{1-\epsilon})$ algorithm for the Max-Weight Increasing Subsequence of length $k$ problem implies a truly subquadratic algorithm for the $(\max,+)$-convolution problem.
\end{proof}

\section{Graphs with Monge Weights}
\label{s:monge}
In this section we study the problem of computing maximum-reward paths with at most $k$ edges in a DAG with edge weights satisfying the Monge property. 
Using the elegant algorithm of \cite{BLP92}, one can compute a single such path in $\widetilde{O}(n)$ time.

\begin{lemma}[From \cite{BLP92}]
Given a DAG with Monge weights,
with $n$ vertices, a pair of vertices $s$ and $t$, and a positive integer $k$, we can compute a maximum reward path from $s$ to $t$ that uses at most $k$ edges, in time 
$\widetilde{O}(n)$. 
\label{mongelinear}
\end{lemma}

Given the adjacency matrix $A$ of the DAG, one can see this equivalently as computing one element of the
matrix power $A^k$ in the tropical semiring (i.e. we replace $(+,\cdot)$ with $(\max,+)$).
Therefore, an important question is whether a whole row or column of $A^k$ can be computed efficiently rather. This corresponds to finding maximum reward paths with $k$ edges from some vertex $s$ to \emph{all} other vertices, or finding maximum reward paths with $k$ edges from some vertex $s$ to some vertex $t$ for all $k$. In Lemma~\ref{khop} we show that one needs $\Omega(n^{3/2})$ time to compute a column of $A^k$ in general.

\begin{lemma}
Given a DAG with Monge weights,
with $n$ vertices, computing the maximum weight path of length $k$ from a given $s$ to all other nodes $t$ requires $\Omega(n^{1.5})$ time.
\label{khop}
\end{lemma}

On the positive side, by further exploiting the Monge property, in Lemma~\ref{forall}
we present an algorithm that can compute any row or column of $A^k$ in $\widetilde{O}(nnz(A)) = \widetilde{O}(m)$ time.

\begin{lemma}
Let $G$ be a DAG of $n$ vertices and $m$ edges equipped with Monge weights that are integers of absolute value at most $M$.
Given a vertex $s$, and a positive integer $k$, we can compute a maximum reward
path from $s$ to $t$ that uses at most $k$ edges, for all $t$, in time 
$O(m\log n \log (nM))$.
Furthermore, if we are given a pair of vertices $s$ and $t$, we can compute as maximum reward path from $s$ to $t$ that uses at most $k$ edges, for all $k\in[n]$, in time 
$O(m\log n \log (nM))$.
\label{forall}
\end{lemma}

\paragraph{Acknowledgments}

We are grateful to Arturs Backurs for insightful discussions that helped us improve this work.

\bibliography{main}

\appendix
\section{Omitted Proofs}

\begin{prevproof}{Lemma}{concavity}
Instead of computing $f(k)$, we will solve the Lagrangian relaxation $\underset{k}{\max}\left\{f(k) - \lambda k\right\}$, where $\lambda \geq 0$ is some parameter.
Note that usually this problem is much easier to solve, since one can incorporate an extra cost of $\lambda$ to the value of each item.
This can be solved by subtracting $\lambda$ from all node weights and computing the maximum-reward path and so can be done in linear time.
Because of the concavity of $f$, this corresponds to finding the uppermost intersection point of the function $f(k)$ with a line with slope $\lambda$.
Let this point be $(k', f(k'))$. Therefore we have computed $f(k')$. If $k>k'$, then we should in fact be looking for smaller $\lambda$, and if $k<k'$ we should be looking for larger $\lambda$.
This yields a binary search algorithm that eventually finds $(k,f(k))$ as the intersection of $f$ with a line with slope $\lambda$.
Therefore the total runtime will be $O(m \log(nM))$.
\end{prevproof}

\begin{prevproof}{Lemma}{khop}
Consider a complete DAG in which for any $i<j$, $w(i,j) = (i-j)^2$, and we wish to minimize the total weight (which is the same as maximizing the total weight if we set weights to $-(i-j)^2$).
To see that these weights are Monge, note that 
\[ (i-j)^2 + (i + 1 - (j+1))^2 - (i - (j+1))^2 - (i + 1 - j)^2  \]
\[= -2 (ij+(i+1)(j+1) - (i+1)j - i(j+1))  \]
\[=  -4 < 0 \]
We will prove that the set of edges used by the solutions has size $\Omega(n^{1.5})$.
Note that the optimal path from $0$ to $n$ with $k$ edges will only contain edges of lengths $\lfloor \frac{n}{k}\rfloor$ and $\lceil \frac{n}{k}\rceil$.
To see this, note that otherwise the path must consist of two edges whose lengths differ by at least $2$. Suppose these lengths are $a$ and $b\geq a + 2$.
But since the cost of an edge depends just on its length, we can replace them by the edges $a+1$ and $b - 1$, thus decreasing the cost by
$a^2 + b^2 - (a+1)^2 - (b-1)^2 = 2 (b - a - 1) > 0$. 
We will suppose wlog that all the larger edges are in the beginning of the path, i.e. closer to $0$ than the smaller ones.

Now, consider any edge $(i,i+x)$. If the following three conditions are true, this edge is part of the solution.
\begin{itemize}
\item{$\lfloor \frac{i}{x}\rfloor \geq x - 1$}
\item{$k(x+1) \leq n$}
\item{$\lceil\frac{i+x}{x}\rceil \leq k$}
\end{itemize}

Suppose that these three conditions are met. We will create an optimal solution containing edge $(i,i+x)$. By the third condition, the path
$a,a+x,a+2x, \dots, i, i+x$, for some $0\leq a < x$ contains at most $k$ edges. Furthermore, since $a < x$, by the first condition we can
increase the lengths of the first $a$ edges by $1$, so that the new path starts at $0$. Finally, extend this path to the right by length-$x$ edges so that it is a $k$-hop path. This can be done
because of the second condition. Therefore we have a path that is optimal for some endpoint and contains $(i,x)$.

This means that the total number of edges is at least
$\sum\limits_{x=1}^{\lfloor\frac{n}{k}\rfloor - 1} \left[ (k-1)x - x(x-1) \right] \geq \Theta(n^{3/2})$,
where we have picked $k = \Theta(\sqrt{n})$ in order to maximize the sum.

Now adding some small amount of arbitrary noise to all the edges ensures that we have to look at all the edges in the solution just to compute the weights of all the solutions.

\end{prevproof}

\begin{prevproof}{Lemma}{forall}

The proof is based on the following lemma.
\begin{lemma}
Consider a DAG with Monge weights, nodes indexed by $0,\dots, n$ in order.
Let $P:=(P_0=0, P_1, \dots, P_{k_1}=n)$ be a maximum reward path from $0$ to $n$,
and $Q:=(Q_0=0, Q_1, \dots, Q_{k_2}=n)$ be a maximum reward path from $0$ to $n$ but
in the DAG where all edge weights are increased by the same positive number (obviously the DAG is still Monge).
There exists a choice of $Q$ such that $Q_1 \geq P_1$.
\label{monotonicity}
\end{lemma}

\begin{proof}
It is easy to see that it suffices to show this for $k_1 = k_2 + 1$, so let $k_2 = k$ and $k_1 = k + 1$.
Suppose that $P_1 > Q_1$.
Such a pair of paths $P, Q$ can be equivalently described as follows:
Let's visualize the paths in their topological order, and scan with a vertical line from left to right, while keeping a point $(x,y)$ in the integer plane, starting from $(0,0)$.
Every time the vertical line meets an endpoint of some edge in $P$ we move from $(x,y)$ to $(x,y+1)$, every time it meets and endpoint of some edge in $Q$ we move to
$(x+1,y)$, and if it meets a common endpoint of both we move to $(x+1,y+1)$. It is easy to see that at any time, $x$ (resp. $y$) is the number of edges of $Q$ (resp. $P$) already encountered
by the
vertical line.
Therefore, since the size of $P$ is $k+1$ and the size of $Q$ is $k$, we will end up at $(k,k+1)$. Now, $P_1 > Q_1$ states the fact that after $(0,0)$ we move to $(1,0)$ and this
means that eventually we will have to cross the line $x=y$.

Furthermore, if we ever move on that line, say e.g. from some $(x,x)$ to $(x+1,x+1)$, 
this means that the paths seen so far are interchangeable and so the path
$Q' := (P_0, \dots, P_{x+1}, Q_{x+2}, \dots, Q_k)$ has the same number of hops and weight as $Q$, but also $Q_1' = P_1$.

Otherwise, at some point we have to move from some $(x,x-1)$ to $(x,x)$ and then to $(x,x+1)$. 
By the Monge property, this means that $Q_x < P_x < P_{x+1} < Q_{x+1}$ and so 
\[ w(Q_x, Q_{x+1}) + w(P_x, P_{x+1}) \geq w(Q_x, P_{x+1}) + w(P_x, Q_{x+1}) \]
So if we define the paths
\[P' := (Q_{0}, \dots, Q_x, P_{x+1}, \dots, P_{k+1}) \] and
\[Q' := (P_{0}, \dots, P_x, Q_{x+1}, \dots, Q_k)\]
by the Monge property and optimality of $P, Q$ we know that $P'$ has the same weight as $P$ and $Q'$ the same weight as $Q$.
Furthermore, $Q'$ also has $Q_1' = P_1$.
\end{proof}

Suppose that we add the number $\lambda$ to all the weights and then find the maximum-weight path that ends at node $n$. This path has $f(\lambda)$ edges for some function 
$f:\mathbb{R}\rightarrow [n]$. It is easy to see that the function $f$ is decreasing and takes all values in $[n]$ for which there exists a path of that length from $0$ to $n$. 
Furthermore, if $w(\lambda)$ is the maximum weight of a path $P$ after adding $\lambda$ to all weights, we know that $w(\lambda) = w(P) + \lambda f(\lambda)$. So among all paths
with $f(\lambda)$ edges, $P$ maximizes $w(\lambda)$, so it also maximizes $w(P)$. Therefore, that path is the maximum weight path with $f(\lambda)$ hops.

Suppose that $a_\lambda(i)$ is the node after $i$ in this path, and that $l_\lambda(i)$ (resp. $r_\lambda(i)$) 
is the number of edges of the form $(i,\star)$ that are shorter (resp. longer) than $(i,a_i)$.

By Lemma~\ref{monotonicity} we know that when looking at paths with $\lambda' > \lambda$, we will have $a_i' \geq a_i$, and when looking at paths with $\lambda' < \lambda$ we will have
$a_i' \leq a_i$. This basically splits our edge set into two subsets, and we can recurse on both of them.
Therefore, we would like to pick $\lambda$ so as to split them as evenly as possible, which we can do by binary search on $\lambda$, each time computing a shortest path on the DAG.
Note that there will always exist a balanced split, since any edge that is not part of a shortest path for any $\lambda$ can be discarded.
Let $r(m)$ be the amount of time the algorithm takes, given a DAG with $m$ edges, plus $n$ edges (one outgoing edge for each vertex). Note that the $n$ edges define an arborescence $A$
that for some choice of $\lambda$ was the shortest path tree of the DAG.
In order to compute the shortest path tree in such a DAG, we run a Breadth-First search using only the $m$ edges, and each time we encounter a path that is shorter than the respective path
in $A$, we update $A$ by substituting an edge $(u,v)$ with another edge $(u',v)$. This takes $O(m)$ time.
Therefore if we denote by $r(m)$ the time the algorithm takes when the number of non-arborescence edges is $m$,
the recursion to (implicitly) compute the shortest path trees for all choices of $\lambda$ 
will be 
\[ r(m) = 2r(m/2) + m\log (nM)\]
so $r(m) = O(m\log n \log (nM))$ and so the total time is $O(m\log n\log (nM))$.

Note that each leaf of this recursion exactly corresponds to an implicit shortest path tree, for a particular value of $\lambda$. In order to reconstruct the shortest path from $0$
to some node $u$ with exactly $k$ edges, we traverse the recursion tree top-down, moving to the left child (smaller $\lambda$) if the number of edges in the current path from $0$ to $u$ is
less than $k$, or to the right child if the number of edges is is more than $k$. 
This way, we can compute the shortest paths of $k$ hops from $0$ to each node, in time $O(m\log n \log (nM))$.
\end{prevproof}

\end{document}